\documentclass[11pt]{article}

\usepackage[a4paper,margin=27mm]{geometry}
\usepackage{amsmath,amssymb,amsthm,mathtools}
\usepackage{booktabs}
\usepackage{graphicx}
\usepackage{microtype}
\usepackage[section]{placeins}
\usepackage[numbers,sort&compress]{natbib}
\usepackage[colorlinks=true,citecolor=blue!55!black,linkcolor=blue!55!black,urlcolor=blue!55!black]{hyperref}
\usepackage{xcolor}

\newtheorem{theorem}{Theorem}[section]
\newtheorem{proposition}[theorem]{Proposition}
\newtheorem{corollary}[theorem]{Corollary}
\newtheorem{lemma}[theorem]{Lemma}
\theoremstyle{definition}

\theoremstyle{remark}

\newcommand{\R}{\mathbb{R}}
\newcommand{\1}{\mathbf{1}}
\newcommand{\diag}{\operatorname{diag}}

\newcommand{\spec}{\operatorname{spec}}
\newcommand{\cF}{\mathcal{F}}
\newcommand{\eps}{\varepsilon}

\title{Collective Hysteresis and Multistability in Threshold Networks}
\author{Moses Boudourides \\[0.5em]
{\small School of Professional Studies, Northwestern University, Evanston, IL, USA} \\
{\small \texttt{Moses.Boudourides@northwestern.edu}}}
\date{}

\begin{document}
\maketitle

\begin{abstract}
In a mechanistic model of the dawn chorus, Kaye showed that heterogeneous activation thresholds and a shared feedback signal determined by the population's active fraction can produce abrupt collective activation and hysteresis. We extend Kaye's mechanism to a network of interacting agents. Each node represents an agent with a continuous activation level, and a nonnegative row-stochastic matrix determines how node activities contribute to each agent's feedback. We prove that weak feedback gives a unique globally attracting equilibrium. For every feedback strength, the homogeneous dynamics reproduce Kaye's scalar equation exactly. Network topology therefore changes neither the folds nor the cusp of this branch, and no heterogeneous mode becomes unstable before the homogeneous mode. In an equitable partition, nodes in the same block receive the same total weight from every block, giving an exact quotient system. With no coupling between blocks, this quotient consists of independent copies of Kaye's scalar equation, and each block can occupy any stable scalar equilibrium. We prove that this combinatorial multistability persists under sufficiently weak interblock coupling. Every assignment of stable scalar states to blocks persists as a stable quotient equilibrium, lifts to a stable full-network equilibrium, and survives small perturbations that break exact equitability. For two symmetrically coupled blocks, branches with unequal block activities terminate at two symmetry-related cusp bifurcations. For feedback strengths just above the onset of bistability in Kaye's scalar model, we derive scaling laws for (i) the interblock coupling at which the heterogeneous cusp bifurcations occur and stable unequal-activity states disappear, (ii) the difference between the activity levels of the two blocks at those bifurcations, and (iii) the displacement of the external stimulus from its value at Kaye's scalar cusp. Computations that track these cusps as feedback varies confirm the scaling laws for gamma, logistic, and normal distributions of individual activation thresholds.
\end{abstract}


\section{Introduction}

Threshold models connect heterogeneous individual decisions to abrupt collective change. In the classical formulation, individuals act once the active fraction exceeds a personal threshold \citep{granovetter1978}; network versions show how local exposure and topology control global cascades \citep{watts2002}. Smooth and stochastic variants additionally produce saddle-node transitions and hysteresis \citep{wiedermann2020,kook2021}. These mechanisms are distinct from networks assembled from prescribed tipping-point normal forms, where the local bistability is assumed before coupling and the principal question is cascade propagation \citep{kroenke2020,klose2020}.

Motivated by the dawn chorus, Kaye recently derived collective hysteresis from a stochastic activation--deactivation process in a population of birds with heterogeneous thresholds \citep{kaye2026}. In the baseline model, every bird receives the same feedback signal, determined by the overall active fraction, rather than feedback transmitted through an explicit interaction network. Under mean-field scaling, the active proportion $x$ obeys
\begin{equation}
 \dot x=-x+F(\alpha u+\beta x),
 \label{eq:kaye-scalar}
\end{equation}
where $u$ is the external stimulus, $\alpha$ is its sensitivity coefficient, $\beta$ is the strength of social feedback, and $F$ is the cumulative distribution function of individual thresholds. A nondegenerate maximum of the threshold density generates a cusp and hence a wedge of bistability. Kaye also develops a spatial continuum extension with normalized interaction kernels and travelling waves. Our setting differs from Kaye's spatial extension: interactions are specified by row-stochastic matrices, and equitable partitions provide exact quotient systems. We use this structure to study network equilibria, multistability, and bifurcations.

Community effects in discrete random threshold networks have been studied numerically, with modularity changing attractor counts and perturbation spreading \citep{wang2013}. Such topological communities, usually characterized by comparatively dense internal connectivity, need not satisfy the equal-input condition required for an equitable partition and exact quotient dynamics. Nathe et al. instead define ``dynamical communities'' through nearly equitable aggregate inputs and detect statistically significant examples in empirical weighted technological, biological, and social networks \citep{nathe2022}.

Exact quotient reductions through equitable partitions are established for consensus, synchronization, and epidemic dynamics \citep{schaub2016,siddique2018,bonaccorsi2015}. Exact symmetry-induced orbit partitions also occur in empirical networks \citep{macarthur2008}, but exact equitability can yield partitions that are too fine for empirical role analysis, motivating $\varepsilon$-equitable and approximate-equitable methods \citep{kate2009,squillace2026}. In studies of Laplacian dynamics, an ``almost equitable partition'' satisfies an exact condition on interblock weights and should not be confused with tolerance-based approximate equitability \citep{schaub2016,bonaccorsi2015}. Weakly coupled cusp systems can support many attractors \citep{izhikevich1998}. These results do not determine how row-stochastic network coupling changes Kaye's model. In particular, they do not show which features of the scalar cusp persist, when weak interblock coupling produces multistability, or how the cusps at which unequal block states disappear approach the scalar cusp.

This paper answers those questions for a deterministic threshold-response network. 
We use multistability to mean the coexistence of several asymptotically stable states at the same parameter values \citep{feudel2008}. Equitable reduction provides the framework for the block-level analysis but is not itself our main result. We first prove that sufficiently weak feedback gives a unique globally attracting equilibrium on every row-stochastic network. We then show that network topology cannot cause a heterogeneous mode to become unstable before the homogeneous mode. Our main multistability result shows that independent assignments of stable scalar equilibria to equitable blocks persist under sufficiently weak interblock coupling, yielding a combinatorial family of stable full-network equilibria. For two symmetrically coupled blocks, we show that branches with unequal block activities end at two symmetry-related cusp bifurcations. Near the onset of scalar bistability, we derive explicit scaling laws for the interblock coupling at which these bifurcations occur, the difference between the two block activities, and the corresponding shift in the external stimulus.

In our deterministic network extension of Kaye's mean-field equation, each node represents an agent with a continuous activation level. Interactions are described by a row-stochastic matrix, so the feedback received by each agent is a weighted average of the node activities. Each activation level responds to this feedback and the external stimulus according to Kaye's cumulative threshold distribution. We do not derive the network equation from stochastic dynamics of binary agents; such a derivation would require a separate limit theorem.

\section{Threshold dynamics on a row-stochastic network}

Consider a network of $N$ interacting agents, with node set $V=\{1,\ldots,N\}$. Node $i$ represents agent $i$ and carries an activation level $x_i\in[0,1]$. The network may be directed and weighted: $w_{ij}$ measures the contribution of agent $j$'s activity to the feedback received by agent $i$. Collecting these weights in $W=(w_{ij})\in\R^{N\times N}$, we assume that $W$ is nonnegative and row-stochastic:
\begin{equation}
 w_{ij}\ge0,
 \qquad
 W\1_N=\1_N.
 \label{eq:row-stochastic}
\end{equation}
Here $\1_N$ denotes the vector of $N$ ones. Thus each row of $W$ sums to one, so every agent receives unit total incoming feedback weight. Let $F:\R\to[0,1]$ be the cumulative distribution function of the activation thresholds. Whenever $F$ is differentiable, we denote its density by $f=F'$. For a constant stimulus $u\in\R$, stimulus sensitivity $\alpha>0$, and feedback strength $\beta\ge0$, the weighted activity contributing to the feedback received by agent $i$ is $(Wx)_i=\sum_{j=1}^N w_{ij}x_j$. The network dynamics are
\begin{equation}
 \dot x=G(x;u,\beta)
 :=-x+\cF(\alpha u\1_N+\beta Wx),
 \qquad x\in[0,1]^N,
 \label{eq:network-model}
\end{equation}
where $\cF$ acts componentwise. Throughout, $\|\cdot\|_\infty$ denotes the infinity norm: for $v\in\R^n$,
$\|v\|_\infty=\max_i|v_i|$, and for $A\in\R^{p\times q}$,
$\|A\|_\infty=\max_i\sum_j|a_{ij}|$.

\begin{proposition}[Well-posed cooperative dynamics]
\label{prop:wellposed}
If $F\in C^1(\R)$, then \eqref{eq:network-model} has a unique global solution for each initial condition in $[0,1]^N$. The cube is forward invariant, at least one equilibrium exists, and the flow is order preserving.
\end{proposition}

\begin{proof}
The vector field is locally Lipschitz. On a face $x_i=0$ one has $G_i=F(\cdot)\ge0$, whereas on $x_i=1$ one has $G_i=-1+F(\cdot)\le0$. Hence the vector field points inward on the boundary of the cube, which is therefore forward invariant. Boundedness on this compact invariant set gives global continuation. The map
$T(x)=\cF(\alpha u\1_N+\beta Wx)$ maps the cube continuously into itself, so Brouwer's theorem gives a fixed point. Finally,
\[
 \frac{\partial G_i}{\partial x_j}
 =\beta f(\alpha u+\beta(Wx)_i)w_{ij}\ge0
 \qquad(i\ne j),
\]
so the Jacobian is Metzler and the flow is cooperative \citep{smith1995}.
\end{proof}

\begin{theorem}[Subcritical contraction]
\label{thm:contraction}
Assume that $F\in C^1(\R)$ and that $L_f:=\sup_{z\in\R}f(z)<\infty$. If $\beta L_f<1$, then \eqref{eq:network-model} has a unique equilibrium $x^*$. Every solution satisfies
\begin{equation}
 \|x(t)-x^*\|_\infty
 \le e^{-(1-\beta L_f)t}\|x(0)-x^*\|_\infty.
 \label{eq:global-contraction}
\end{equation}
This conclusion requires neither irreducibility nor the existence of an equitable partition.
\end{theorem}

\begin{proof}
Because $\|W\|_\infty=1$, the equilibrium map
$T(x)=\cF(\alpha u\1_N+\beta Wx)$ is a contraction of $[0,1]^N$ in the infinity norm, with Lipschitz constant at most $q=\beta L_f<1$. Banach's theorem gives a unique fixed point $x^*$. Variation of constants gives
\[
 x(t)-x^*=e^{-t}[x(0)-x^*]
 +\int_0^t e^{-(t-s)}\{T(x(s))-T(x^*)\}\,ds.
\]
Taking infinity norms, multiplying by $e^t$, and applying Gronwall's inequality yields \eqref{eq:global-contraction}.
\end{proof}

\subsection{Equitable partitions and exact quotients}

Let $\mathcal P=\{C_1,\ldots,C_m\}$ be a partition of the network nodes into nonempty blocks. Its indicator matrix $H\in\{0,1\}^{N\times m}$ is defined by $H_{ia}=1$ if and only if $i\in C_a$. Thus $H\1_m=\1_N$ and $H$ has full column rank.

\begin{lemma}[Characterization of equitable partitions]
\label{lem:equitable}
The partition $\mathcal P$ is equitable for $W$, meaning that for every pair of blocks $C_a,C_b$ the sum
\begin{equation}
 \sum_{j\in C_b}w_{ij}
 \label{eq:block-input}
\end{equation}
is independent of the choice of $i\in C_a$, if and only if there exists an $m\times m$ matrix $B$ such that
\begin{equation}
 WH=HB.
 \label{eq:equitable}
\end{equation}
The matrix $B$ is unique and is given by
\begin{equation}
 B=(H^TH)^{-1}H^TWH.
 \label{eq:quotient-matrix}
\end{equation}

\noindent In that case,
\[
 B_{ab}=\sum_{j\in C_b}w_{ij}
 \qquad \text{for any } i\in C_a,
\]
so $B_{ab}$ is the total weight received by a node in block $C_a$ from the nodes in block $C_b$. If $W$ is nonnegative and row-stochastic, then so is $B$. This is the standard adjacency-matrix characterization of an equitable partition and its quotient matrix \citep{godsil2001,schaub2016}.
\end{lemma}

\begin{proof}
For $i\in C_a$,
\[
 (WH)_{ib}=\sum_{j\in C_b}w_{ij},
 \qquad
 (HB)_{ib}=B_{ab}.
\]
Hence \eqref{eq:equitable} holds if and only if \eqref{eq:block-input} is constant over $i\in C_a$ for every $a,b$. Because $H$ has full column rank, multiplying \eqref{eq:equitable} by the left inverse $(H^TH)^{-1}H^T$ gives \eqref{eq:quotient-matrix}, proving uniqueness. Nonnegativity of $B$ follows from its block-sum interpretation. Finally,
\[
 HB\1_m=WH\1_m=W\1_N=\1_N=H\1_m,
\]
and full column rank of $H$ implies $B\1_m=\1_m$.
\end{proof}

\begin{theorem}[Exact quotient reduction]
\label{thm:quotient}
Suppose $\mathcal P$ is equitable for $W$, with quotient matrix $B$ from Lemma \ref{lem:equitable}. Then the block-constant subspace
\[
 \mathcal S_H=\{Hy:y\in\R^m\}
\]
is invariant under \eqref{eq:network-model}, and the restricted dynamics are exactly
\begin{equation}
 \dot y=-y+\cF(\alpha u\1_m+\beta By).
 \label{eq:quotient-model}
\end{equation}
\end{theorem}

\begin{proof}
If $x=Hy$, then $WHy=HBy$. Componentwise application of the same scalar function commutes with replication by $H$, namely $\cF(Hv)=H\cF(v)$. Therefore
\[
 G(Hy;u,\beta)
 =H\{-y+\cF(\alpha u\1_m+\beta By)\}\in\mathcal S_H,
\]
which proves invariance and the quotient equation.
\end{proof}

The blocks of an equitable partition are input-equivalence classes: nodes in the same block receive the same aggregate weight from every block. They need not be communities in the usual modular sense of dense within-group and sparse between-group connectivity \citep{nathe2022}.

Lemma \ref{lem:equitable} and Theorem \ref{thm:quotient} are structural: they apply to directed or weighted networks and require neither symmetry nor equal block sizes. Figure \ref{fig:quotient} illustrates the exact reduction on a fifteen-node network with three unequal equitable blocks. Full-network trajectories initialized in $\mathcal S_H$ agree with the lifted quotient trajectories to numerical integration tolerance.

\begin{figure}[t]
 \centering
 \includegraphics[width=0.68\linewidth]{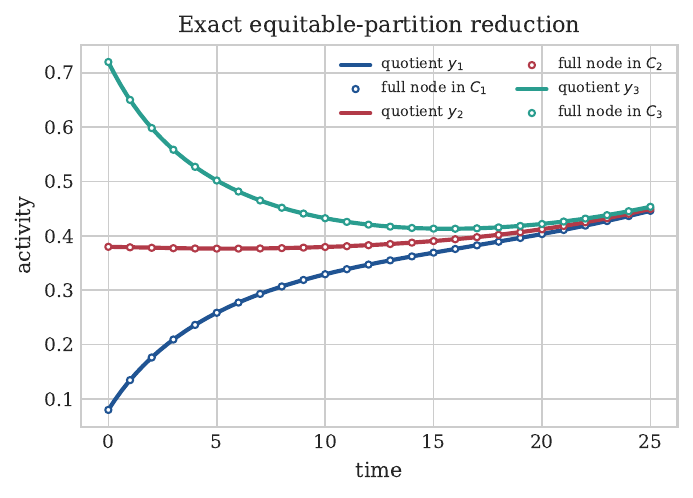}
 \caption{Exact quotient dynamics on a network with block sizes $4$, $5$, and $6$. Curves solve the three-dimensional quotient system; open circles show one representative node from each block in the fifteen-dimensional system. The computed identity error $\|WH-HB\|_\infty$ is $6.94\times10^{-18}$ and the maximum trajectory discrepancy is $3.65\times10^{-11}$.}
 \label{fig:quotient}

 \smallskip
 \noindent\textbf{Alt text:} Three curves show the quotient trajectories for a fifteen-node network with equitable blocks of sizes 4, 5, and 6. Open circles show the corresponding full-network trajectory for one representative node from each block and lie on top of the quotient curves, illustrating exact quotient reduction.
\end{figure}

\section{Homogeneous dynamics and spectral stability}

Because $W\1_N=\1_N$, the one-block partition is equitable for every row-stochastic $W$, so Theorem \ref{thm:quotient} applies to the homogeneous subspace.

\begin{corollary}[Scalar embedding]
\label{cor:scalar}
For every row-stochastic $W$, the homogeneous set $\{s\1_N:s\in[0,1]\}$ is invariant and carries exactly the scalar equation \eqref{eq:kaye-scalar}. Hence its equilibrium branch, folds, and cusp are independent of network topology.
\end{corollary}

A homogeneous equilibrium has the form $x^*=s^*\1_N$, where
\begin{equation}
 s^*=F(z^*),
 \qquad
 z^*=\alpha u+\beta s^*.
 \label{eq:hom-eq}
\end{equation}
Its Jacobian is
\begin{equation}
 J_*=-I_N+aW,
 \qquad
 a=\beta f(z^*).
 \label{eq:hom-jacobian}
\end{equation}

A perturbation proportional to $\1_N$ is homogeneous. We call an eigenmode of $J_*$ that is not proportional to $\1_N$ heterogeneous.

\begin{theorem}[Spectral ordering at the homogeneous branch]
\label{thm:spectral}
Let $W$ be row-stochastic. A homogeneous equilibrium is asymptotically stable if $a<1$ and unstable if $a>1$. At $a=1$, zero is semisimple and
\[
 \ker J_*=\ker(I_N-W).
\]
If $W$ is irreducible, this kernel is $\operatorname{span}\{\1_N\}$, so zero is simple and every heterogeneous mode has strictly negative real part. These conclusions do not require aperiodicity. If $W$ is reducible, zero can have multiplicity greater than one and heterogeneous neutral directions can occur at the scalar fold, but no mode becomes unstable before it.
\end{theorem}

\begin{proof}
If $\lambda\in\spec(W)$, then $-1+a\lambda\in\spec(J_*)$. Every eigenvalue of a row-stochastic matrix satisfies $|\lambda|\le1$, while $W\1_N=\1_N$. Hence, if $a<1$,
$\operatorname{Re}(-1+a\lambda)\le-1+a<0$; if $a>1$, the eigenvalue $-1+a$ is positive. At $a=1$, the zero eigenspace is $\ker(I_N-W)$. The eigenvalue $1$ of $W$ is semisimple: otherwise a nontrivial Jordan block would make $\|W^n\|_\infty$ grow without bound, contradicting $\|W^n\|_\infty=1$. If $W$ is irreducible, Perron--Frobenius theory makes the eigenvalue $1$ simple \citep{seneta2006}. Periodicity may place other eigenvalues on the unit circle, but any such $\lambda\ne1$ has $\operatorname{Re}\lambda<1$, so its Jacobian eigenvalue still has negative real part.
\end{proof}

Thus, under row-stochastic coupling, no heterogeneous mode becomes unstable before the scalar fold. Irreducibility is needed only to make the critical direction uniquely homogeneous. If the row-stochastic normalization is removed, the homogeneous set need not remain invariant when the row sums differ, and both the scalar reduction and the spectral ordering can fail.

For a smooth density with a strict nondegenerate maximum at $z_c$,
\begin{equation}
 f(z_c)>0,
 \qquad f'(z_c)=0,
 \qquad f''(z_c)<0,
 \label{eq:cusp-density}
\end{equation}
the scalar cusp occurs at
\begin{equation}
 x_c=F(z_c),
 \qquad
 \beta_c=\frac{1}{f(z_c)},
 \qquad
 u_c=\frac{z_c-\beta_c x_c}{\alpha}.
 \label{eq:scalar-cusp}
\end{equation}
These are the standard scalar cusp conditions \citep{kaye2026,kuznetsov2004}. Figure \ref{fig:scalar} shows the benchmark used below.

\begin{figure}[t]
 \centering
 \includegraphics[width=\linewidth]{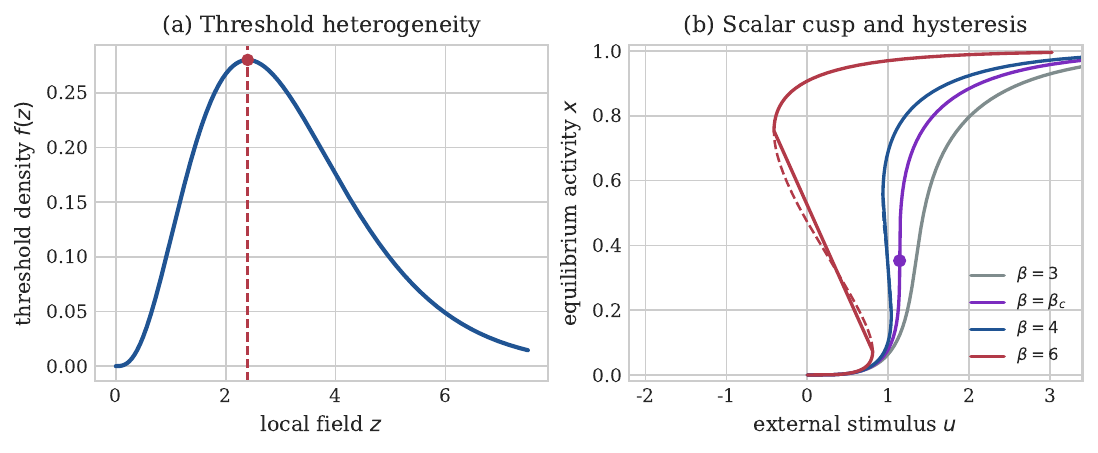}
 \caption{Scalar benchmark for a Gamma threshold distribution with shape $4$ and scale $0.8$. (a) The density has a nondegenerate maximum at $z_c=2.4$. (b) At $\beta=\beta_c=3.570762$, the two folds meet at the cusp; for $\beta>\beta_c$, the branch contains a bistable interval and exhibits hysteresis. Solid and dashed segments are stable and unstable, respectively; the dot marks the cusp.}
 \label{fig:scalar}

  \smallskip
 \noindent\textbf{Alt text:} Panel (a) shows a unimodal gamma threshold density as a function of the local field, with its maximum marked by a dot and a vertical dashed line. Panel (b) shows scalar equilibrium activity versus external stimulus for several feedback strengths; a cusp marks the onset of bistability, and solid and dashed branches indicate stable and unstable equilibria.
\end{figure}

\section{Combinatorial multistability under weak interblock coupling}

Consider a $C^1$ family of nonnegative row-stochastic quotient matrices $B(\eps)$ satisfying
\begin{equation}
 B(\eps)=I_m+\eps L+O(\eps^2),
 \qquad \eps\downarrow0.
 \label{eq:weak-family}
\end{equation}
Here $\eps$ measures the strength of interblock coupling. Since $B(\eps)\1_m=\1_m$, we have $L\1_m=0$. At $\eps=0$, $B(0)=I_m$, and the quotient system consists of $m$ independent copies of Kaye's scalar equation.

\begin{theorem}[Persistence and counting]
\label{thm:persistence}
Fix $(u,\beta)$ and suppose the scalar equation \eqref{eq:kaye-scalar} has $r$ distinct hyperbolic asymptotically stable equilibria $s_1,\ldots,s_r$. Each assignment $\sigma=(\sigma_1,\ldots,\sigma_m)\in\{1,\ldots,r\}^m$ defines the product equilibrium
\[
 y^0_\sigma=(s_{\sigma_1},\ldots,s_{\sigma_m}).
\]
There exists $\eps_0>0$ such that every product equilibrium extends uniquely to a $C^1$ branch $y_\sigma(\eps)$ of asymptotically stable equilibria for $0\le\eps<\eps_0$. Hence the quotient system has at least $r^m$ stable equilibria throughout this interval. In a scalar bistable interval, $r=2$, so the count is at least $2^m$.
\end{theorem}

\begin{proof}
At $\eps=0$, the equilibrium equation separates. Its state Jacobian at $y^0_\sigma$ is diagonal, with entries
\[
 d_a=-1+\beta f(\alpha u+\beta s_{\sigma_a})<0.
\]
It is therefore invertible. The implicit-function theorem gives a unique local equilibrium branch through each product state. Eigenvalues depend continuously on matrix entries, so the branches remain asymptotically stable for sufficiently small $\eps$. Since the product equilibria are distinct, disjoint neighbourhoods can be chosen before applying the theorem, preserving the count.
\end{proof}

\begin{proposition}[Full-network stability of a lifted quotient equilibrium]
\label{prop:lift-stability}
Let $WH=HB$, and let $y^*$ be an equilibrium of the quotient system \eqref{eq:quotient-model}. Define
\[
 q_a=f\!\left(\alpha u+\beta(By^*)_a\right).
\]
Then $x^*=Hy^*$ is an equilibrium of the full network. If
\begin{equation}
 \beta\max_{1\le a\le m}q_a<1,
 \label{eq:transverse-condition}
\end{equation}
then $x^*$ is asymptotically stable against arbitrary perturbations in $\R^N$, including perturbations transverse to $\mathcal S_H$.
\end{proposition}

\begin{proof}
The equilibrium statement follows from Theorem \ref{thm:quotient}. At the lift,
\[
 J_N=-I_N+\beta D_*W,
 \qquad
 D_*=\diag(q_{c(1)},\ldots,q_{c(N)}),
\]
where $c(i)$ is the block containing node $i$. Since $D_*W$ is nonnegative and its $i$th row sum is $q_{c(i)}$,
$\rho(\beta D_*W)\le\|\beta D_*W\|_\infty=\beta\max_a q_a<1$. Every eigenvalue of $J_N$ therefore has negative real part.
\end{proof}

\begin{corollary}[Full-network form of combinatorial multistability]
\label{cor:full-persistence}
Suppose a family $W(\eps)$ has a fixed equitable partition $H$ and quotient $B(\eps)$ satisfying \eqref{eq:weak-family}. Under the assumptions of Theorem \ref{thm:persistence}, the $r^m$ lifted equilibria $Hy_\sigma(\eps)$ are asymptotically stable in the full $N$-node system for all sufficiently small $\eps\ge0$.
\end{corollary}

\begin{proof}
For $j=1,\ldots,r$, let $z_j=\alpha u+\beta s_j$. At $\eps=0$, hyperbolic stability of each selected scalar equilibrium gives
$\beta f(z_{\sigma_a})<1$ for every block. 
The quotient branches and their local fields depend continuously on $\eps$, so \eqref{eq:transverse-condition} holds uniformly over the finitely many branches for sufficiently small $\eps$. Proposition \ref{prop:lift-stability} completes the proof.
\end{proof}

\begin{proposition}[Persistence after breaking equitability]
\label{prop:approx-equitable}
Let $W_0$ be row-stochastic and equitable for $H$, and let $x_0=Hy^*$ be a hyperbolic equilibrium of the full network. For every row-stochastic $W$ with $\|W-W_0\|_\infty$ sufficiently small, there is a unique equilibrium $x^*(W)$ near $x_0$, depending $C^1$-smoothly on $W$. If $x_0$ is asymptotically stable, then so is $x^*(W)$, and
\begin{equation}
 \|x^*(W)-x_0\|_\infty=O(\|W-W_0\|_\infty).
 \label{eq:approx-equitable-bound}
\end{equation}
Exact equality of node activities within blocks need not survive when $WH=HB$ is broken.
\end{proposition}

\begin{proof}
Apply the implicit-function theorem to $G(x;u,\beta,W)=0$ at $(x_0,W_0)$. Hyperbolicity makes $D_xG(x_0;u,\beta,W_0)$ invertible and gives the local $C^1$ equilibrium branch and estimate \eqref{eq:approx-equitable-bound}. Stability follows from continuity of the spectrum. Without equitability, the vector field need not preserve $\mathcal S_H$, so the continued equilibrium is generally not block constant.
\end{proof}

The persistence theorem therefore counts full-network attractors, not merely attractors within the invariant quotient. Proposition \ref{prop:approx-equitable} shows that each hyperbolic block state also survives sufficiently small non-equitable perturbations, although its node activities are no longer exactly equal within each block. This is a local continuation result near an exactly equitable matrix; it does not construct an approximately invariant quotient or provide a global reduction error bound. To determine how each product equilibrium initially changes when interblock coupling is introduced, we differentiate the quotient equilibrium equation at $\eps=0$.

\begin{proposition}[First-order displacement]
\label{prop:displacement}
Under the assumptions of Theorem \ref{thm:persistence}, define
$z_j=\alpha u+\beta s_j$. Then
\begin{equation}
 [y'_\sigma(0)]_a
 =\frac{\beta f(z_{\sigma_a})}{1-\beta f(z_{\sigma_a})}
 (Ly^0_\sigma)_a.
 \label{eq:first-displacement}
\end{equation}
\end{proposition}

\begin{proof}
Differentiate
$-y+\cF(\alpha u\1_m+\beta B(\eps)y)=0$
with respect to $\eps$ at zero and solve the resulting diagonal linear system. The $a$th equation is
\[
 [-1+\beta f(z_{\sigma_a})] [y'_\sigma(0)]_a
 +\beta f(z_{\sigma_a})(Ly^0_\sigma)_a=0,
\]
which is equivalent to \eqref{eq:first-displacement}.
\end{proof}

For two symmetrically coupled blocks, $L=\left(\begin{smallmatrix}-1&1\\1&-1\end{smallmatrix}\right)$. As $\eps$ increases from zero, a low--high state initially moves inward: its low component increases and its high component decreases. Figure \ref{fig:branches} tracks this contraction and the subsequent loss of stable mixed branches.

\begin{figure}[t]
 \centering
 \includegraphics[width=\linewidth]{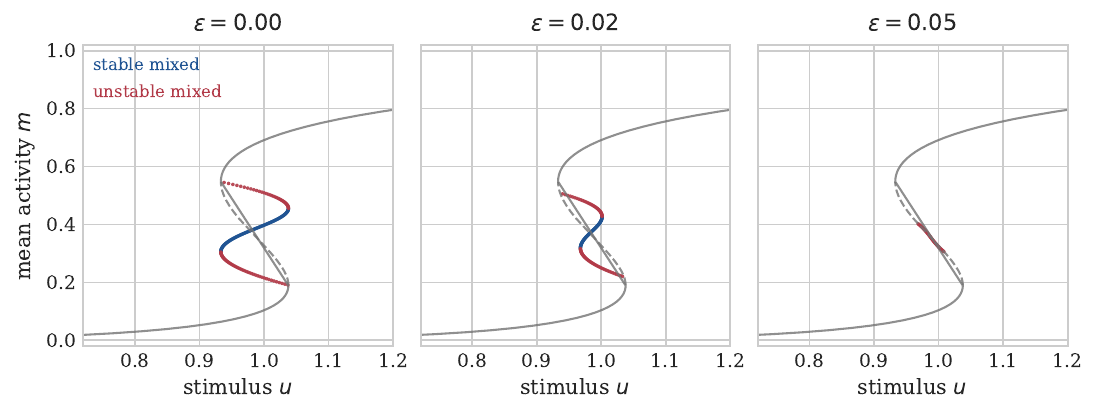}
 \caption{Two-block equilibria for the Gamma benchmark at $\beta=4$. Grey curves are the topology-independent homogeneous branch. Blue and red points are stable and unstable heterogeneous equilibria, respectively; exchanging blocks gives the reflected state with the same mean. Weak coupling compresses and narrows the stable mixed branch. At $\eps=0.05$, above the heterogeneous-cusp value, no stable mixed equilibrium remains.}
 \label{fig:branches}

  \smallskip
 \noindent\textbf{Alt text:} Three panels show two-block equilibrium branches at increasing interblock coupling. Grey curves give the homogeneous branch, while colored branches show stable and unstable mixed equilibria; as coupling increases, the mixed branch narrows and then disappears.
\end{figure}

\section{Heterogeneous cusps for two symmetrically coupled blocks}

Set
\begin{equation}
 B_\eps=
 \begin{pmatrix}
 1-\eps&\eps\\
 \eps&1-\eps
 \end{pmatrix},
 \qquad 0\le\eps\le\frac12.
 \label{eq:two-B}
\end{equation}

Here $\eps$ is the fraction of the total feedback weight that each block receives from the other block. The local fields at an equilibrium are
\[
 z_1=\alpha u+\beta\bigl((1-\eps)x_1+\eps x_2\bigr),
 \qquad
 z_2=\alpha u+\beta\bigl(\eps x_1+(1-\eps)x_2\bigr).
\]

The equilibrium equations are invariant under exchange of the two blocks. Introduce mean and contrast coordinates
\[
 m=\frac{x_1+x_2}{2},
\qquad
d=\frac{x_1-x_2}{2},
\]
and local fields
\[
 a=\frac{z_1+z_2}{2},
 \qquad
 s=\frac{z_1-z_2}{2}.
\]
For heterogeneous equilibria $s\ne0$, define
\begin{equation}
 M(a,s)=\frac{F(a+s)+F(a-s)}{2},
 \qquad
 D(a,s)=\frac{F(a+s)-F(a-s)}{2}.
 \label{eq:MD}
\end{equation}
For each heterogeneous equilibrium, $m=M(a,s)$ and $d=D(a,s)$. The corresponding stimulus $u$ and interblock coupling $\eps$ are given by
\begin{equation}
 U(a,s;\beta)=\frac{a-\beta M(a,s)}{\alpha},
 \qquad
 E(a,s;\beta)=\frac12\left(1-\frac{s}{\beta D(a,s)}\right).
 \label{eq:parameter-map}
\end{equation}
Thus $u=U(a,s;\beta)$ and $\eps=E(a,s;\beta)$, and the map
$P_\beta=(U,E):(a,s)\mapsto(u,\eps)$ parameterizes the heterogeneous equilibria. Folds are singular points of $P_\beta$, and a terminal mixed-state cusp is a Whitney cusp of this map.

\begin{theorem}[Heterogeneous cusps born from the scalar cusp]
\label{thm:heterocusp}
Assume $F\in C^7$ near $z_c$ and \eqref{eq:cusp-density}. There is $\eta>0$ such that, for every
$\beta\in(\beta_c,\beta_c+\eta)$, the parameter map $P_\beta$ has two exchange-related generic Whitney cusp points. They correspond to heterogeneous equilibria with opposite values of $d$ and map to the same parameter pair $(u_c(\beta),\eps_c(\beta))$.

Writing $\delta=\beta-\beta_c$, their asymptotics are
\begin{align}
 \eps_c(\beta)
 &=\frac{\delta}{3\beta_c}+o(\delta),
 \label{eq:eps-scaling}\\
 |x_{1,c}-x_{2,c}|
 &=2\left[-\frac{2}{f''(z_c)\beta_c^4}\right]^{1/2}
 \delta^{1/2}+o(\delta^{1/2}),
 \label{eq:contrast-scaling}\\
 u_c(\beta)-u_c
 &=-\frac{x_c}{\alpha}\delta+o(\delta).
 \label{eq:u-scaling}
\end{align}
\end{theorem}

Thus, as $\beta\downarrow\beta_c$, the interblock coupling at the heterogeneous cusps and the displacement of the stimulus both vanish linearly in $\beta-\beta_c$, whereas the difference between the two block activities vanishes as $(\beta-\beta_c)^{1/2}$.

\begin{proof}
Let
\[
 \Phi(a,s;\beta)=\det D_{(a,s)}P_\beta(a,s).
\]
At a fold, a kernel vector of $DP_\beta$ is
$k=(E_s,-E_a)^T$. The cusp equations are
\begin{equation}
 \Phi=0,
 \qquad
 \Psi:=D_k\Phi=0.
 \label{eq:cusp-equations}
\end{equation}
The functions in \eqref{eq:parameter-map} have removable singularities at $s=0$ after writing $D(a,s)/s$ as its even Taylor expansion.

Put
\begin{equation}
 \beta=\beta_c+t^2,
 \qquad
 a=z_c+t^2A,
 \qquad
 s=tS.
 \label{eq:blowup}
\end{equation}
Direct Taylor expansion of \eqref{eq:cusp-equations}, followed by division by nonzero smooth factors and the lowest powers of $t$, produces a smooth equivalent system
\begin{align}
 R_1(A,S,t)
 &=\frac{1}{\beta_c}+\frac{\beta_c f''(z_c)}{2}S^2+O(t)=0,
 \label{eq:R1}\\
 R_2(A,S,t)
 &=A+\frac{f'''(z_c)}{8f''(z_c)}S^2+O(t)=0.
 \label{eq:R2}
\end{align}
At $t=0$, the two solutions are
\begin{equation}
 S_\pm
 =\pm\left[-\frac{2}{\beta_c^2f''(z_c)}\right]^{1/2},
 \qquad
 A_\pm=-\frac{f'''(z_c)}{8f''(z_c)}S_\pm^2.
 \label{eq:leading-AS}
\end{equation}
The Jacobian of $(R_1,R_2)$ with respect to $(S,A)$ has determinant
$\beta_c f''(z_c)S_\pm\ne0$. The implicit-function theorem therefore gives two branches for all sufficiently small $t>0$. They are exchanged by $s\mapsto-s$.

The singularity has rank one for small $t>0$. To test its order invariantly, compose the parameter map with the blow-up transformation in \eqref{eq:blowup}, writing $\widetilde P_t(A,S)=P_{\beta_c+t^2}(z_c+t^2A,tS)$ and $\widetilde\Phi_t=\det D\widetilde P_t$. If $\widetilde k=(E_S,-E_A)^T$ spans the kernel of $D\widetilde P_t$ along the fold, direct expansion on the cusp branch gives
\begin{equation}
 D_{\widetilde k}^2\widetilde\Phi_t
 =\frac{5\beta_c^4[f''(z_c)]^4S_\pm^3}{216\alpha}t^{12}+O(t^{13}),
 \label{eq:genericity-expansion}
\end{equation}
which is nonzero. The same expansion gives $\partial_S\widetilde\Phi_t=[f''(z_c)/(3\alpha)]t^6+O(t^7)$ on the cusp branch, so the fold set is regular. The blow-up coordinate change is invertible for each $t>0$, so cusp type is preserved. Together with rank one and first directional degeneracy, \eqref{eq:genericity-expansion} is the Whitney-cusp nondegeneracy condition \citep{kuznetsov2004}. Hence both branches consist of generic cusps.

It remains to recover the parameter values and equilibrium activities. From \eqref{eq:MD},
\[
 \frac{D(a,s)}{s}
 =f(z_c)+\frac{f''(z_c)}{6}s^2+o(t^2).
\]
Using \eqref{eq:R1} and $f(z_c)=1/\beta_c$ in \eqref{eq:parameter-map} yields
$\eps_c=t^2/(3\beta_c)+o(t^2)$. Also
$d=D(a,s)=s/\beta_c+o(t)$, which gives \eqref{eq:contrast-scaling}. Finally,
$M(a,s)=x_c+f(z_c)t^2A+o(t^2)$, and the $A$ terms cancel in $a-\beta M$, giving
$U-u_c=-(x_c/\alpha)t^2+o(t^2)$. Since $t^2=\delta$, \eqref{eq:eps-scaling}--\eqref{eq:u-scaling} follow.
\end{proof}

These cusp bifurcations do not lie on the homogeneous branch. They mark the disappearance of stable unequal-activity states as interblock coupling increases and converge to the scalar cusp as $\beta\downarrow\beta_c$.

\section{Numerical results}
\label{sec:numerics}

The computations were deterministic. Whenever possible, we parameterized equilibrium branches by the local fields, which allowed us to trace them through folds. We determined stability from the eigenvalues of the exact Jacobian. We located the heterogeneous cusp by solving the fold and cusp conditions and checked its genericity independently using Lyapunov--Schmidt coefficients. We also used symbolic algebra to verify the Taylor-expansion coefficients used to derive the three scaling laws in Theorem \ref{thm:heterocusp}.

We used a gamma distribution with shape $4$ and scale $0.8$ as the primary benchmark and set $\alpha=1$. This is the same smooth unimodal family used to illustrate threshold heterogeneity in the scalar setting. Table \ref{tab:numerics} reports the scalar cusp and the two-block cusp computed at $\beta=4$.

\begin{table}[t]
\centering
\caption{Numerical benchmark for the Gamma threshold law.}
\label{tab:numerics}
\begin{tabular}{@{}lll@{}}
\toprule
Quantity & Value & Role \\
\midrule
$z_c$ & $2.400000$ & density mode \\
$x_c$ & $0.352768$ & scalar cusp activity \\
$\beta_c$ & $3.570762$ & scalar critical feedback \\
$u_c$ & $1.140349$ & scalar critical stimulus \\
$\eps_c$ at $\beta=4$ & $0.036750$ & mixed-state cusp coupling \\
$u_c(4)$ & $0.985362$ & mixed-state cusp stimulus \\
$(x_{1,c},x_{2,c})$ & $(0.542897,0.186896)$ & cusp block activities \\
Cubic coefficient & $32.762588$ & cusp nondegeneracy \\
Unfolding determinant & $2.355585$ & parameter transversality \\
\bottomrule
\end{tabular}
\end{table}

Figure \ref{fig:foldset} shows the two fold curves bounding the mixed equilibria at $\beta=4$. The curves meet at a heterogeneous cusp of the type described in Theorem \ref{thm:heterocusp}. At the computed cusp, we obtained an equilibrium residual of $3.33\times10^{-16}$, a Jacobian determinant of magnitude $2.36\times10^{-17}$, and a quadratic cusp coefficient of magnitude $1.98\times10^{-15}$. The nonzero cubic coefficient and nonsingular two-parameter unfolding reported in Table \ref{tab:numerics} provided an independent numerical check of genericity.

\begin{figure}[t]
 \centering
 \includegraphics[width=0.62\linewidth]{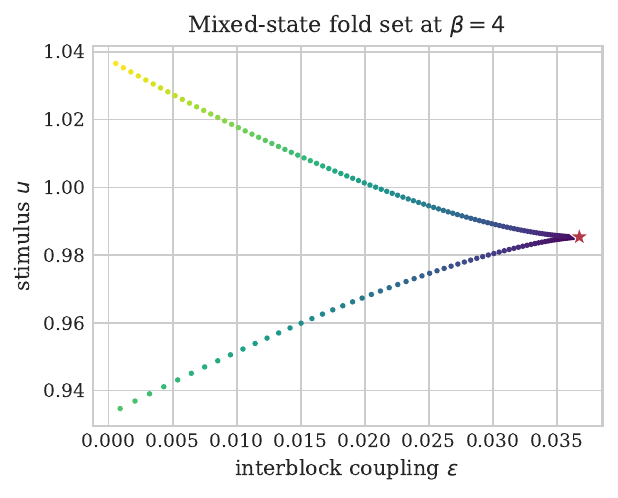}
 \caption{Mixed-state fold set for the Gamma benchmark at $\beta=4$. Colour records block contrast $|x_1-x_2|$; the star is the generic heterogeneous cusp at $(\eps,u)=(0.036750,0.985362)$.}
 \label{fig:foldset}

 \smallskip
 \noindent\textbf{Alt text:} A two-curve fold set is plotted in the plane of interblock coupling and stimulus, with color indicating the difference between the two block activities. The two fold curves meet at a starred heterogeneous cusp.
\end{figure}

To test the asymptotic coefficients as well as the exponents, we tracked the cusp over the range $10^{-5}\le(\beta-\beta_c)/\beta_c\le0.3$. Figure \ref{fig:scaling} shows each computed quantity divided by its leading asymptotic formula; all three ratios approach one as $\beta\downarrow\beta_c$. Log--log fits over relative distances no greater than $0.02$ gave exponents $0.9989$, $0.4991$, and $1.0003$ for the gamma distribution. The logistic and normal threshold distributions gave the same limiting exponents within $1.3\times10^{-3}$.

\begin{figure}[t]
 \centering
 \includegraphics[width=\linewidth]{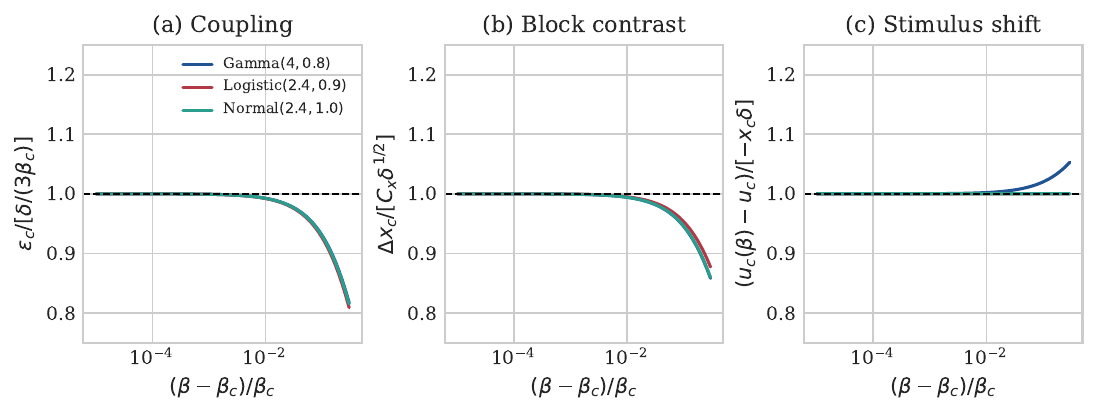}
 \caption{Coefficient-level tests of Theorem \ref{thm:heterocusp}. Each ordinate is the computed value divided by the corresponding leading asymptotic formula. The dashed line is one. Here $\Delta x_c=|x_{1,c}-x_{2,c}|$ and $C_x=2[-2/(f''(z_c)\beta_c^4)]^{1/2}$.}
 \label{fig:scaling}

 \smallskip
 \noindent\textbf{Alt text:} Three panels show ratios between computed cusp quantities and their leading asymptotic formulas as feedback approaches the scalar critical value. The ratios for coupling, block contrast, and stimulus shift all approach one for gamma, logistic, and normal threshold distributions.
\end{figure}

Figure \ref{fig:counts} shows the results for a symmetric family in which the total interblock weight $\eps$ was distributed equally among the other blocks. At $\eps=0$, scalar bistability gives exactly $2^m$ stable product equilibria. We tracked these equilibria as $\eps$ increased and recorded the successive losses of stable branches. For $m=3$, we also followed all eight branches down to $\eps=10^{-4}$; every branch was stable, and the maximum error in the derivative predicted by \eqref{eq:first-displacement} decreased linearly with $\eps$. Theorem \ref{thm:persistence} guarantees the count only for sufficiently weak coupling. The computations at larger values of $\eps$ illustrate the loss of the tracked branches but do not constitute an exhaustive search for all equilibria.

\begin{figure}[t]
 \centering
 \includegraphics[width=0.68\linewidth]{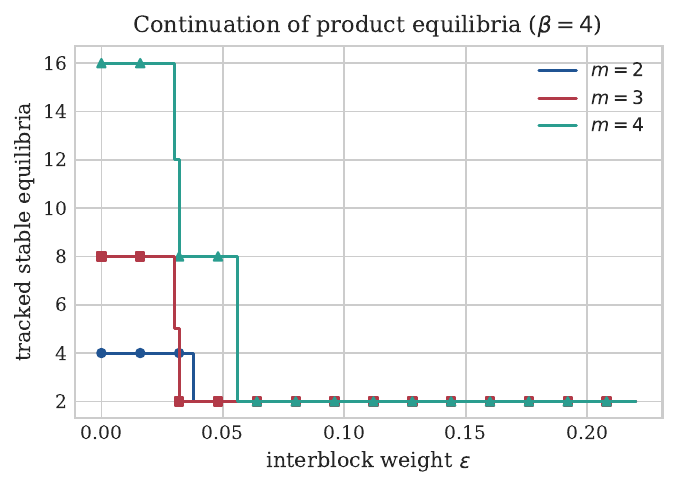}
 \caption{Stable product-state branches tracked for $m=2,3,4$ blocks at $\beta=4$ and the midpoint of the scalar bistable interval. The initial counts are $2^m$, as guaranteed by Theorem \ref{thm:persistence}; stronger interblock coupling removes heterogeneous branches until only the two homogeneous attractors remain among the tracked states.}
 \label{fig:counts}

 \smallskip
 \noindent\textbf{Alt text:} A step plot shows the number of tracked stable product equilibria versus interblock coupling for systems with 2, 3, and 4 blocks. Each curve starts at $2^m$ stable equilibria when coupling is zero and then drops in steps until only two homogeneous equilibria remain.
\end{figure}

To examine departures from exact equitability, we perturbed the fifteen-node equitable matrix used in Figure \ref{fig:quotient} toward a fixed strictly positive row-stochastic matrix. The initial low--high--low block equilibrium was hyperbolic and stable. Figure \ref{fig:approx-equitable} shows its displacement against $\|W-W_0\|_\infty$ and its departure from block constancy against the equitability defect $\|WH-H\widehat B\|_\infty$, where $\widehat B=(H^TH)^{-1}H^TWH$ is the least-squares quotient. Both relationships were linear to leading order; fits over the six smallest perturbations gave slopes $1.034$ and $0.988$, respectively. All computed equilibria remained stable. These computations illustrate the local persistence result in Proposition \ref{prop:approx-equitable}; they do not establish an approximate quotient or a global robustness threshold.

\begin{figure}[t]
 \centering
 \includegraphics[width=\linewidth]{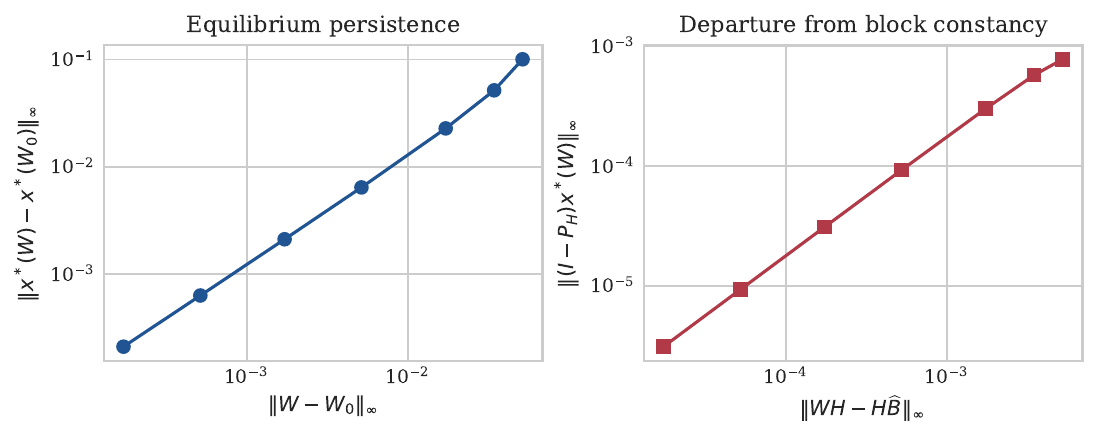}
\caption{Persistence of an equilibrium under deterministic perturbations of a fifteen-node network that initially has three equitable blocks. (a) The equilibrium displacement is linear to leading order in the matrix perturbation. (b) Within-block dispersion is linear to leading order in the equitability defect; $P_H=H(H^TH)^{-1}H^T$ is the block-constant projector.}
 \label{fig:approx-equitable}

 \smallskip
 \noindent\textbf{Alt text:} Panel (a) plots equilibrium displacement against the size of a matrix perturbation on logarithmic axes and shows an approximately linear trend. Panel (b) plots within-block dispersion against the equitability defect, also on logarithmic axes, and again shows an approximately linear trend, illustrating persistence under small non-equitable perturbations.
\end{figure}

\section{Discussion}

The threshold distribution and the network structure play different roles. If $\beta\sup_z f(z)<1$, every row-stochastic network has a unique globally attracting equilibrium. At stronger feedback, the threshold distribution determines the scalar cusp, while the curvature of its density at the maximum determines the local scaling near that cusp. Row-stochastic coupling leaves the homogeneous branch unchanged and cannot make a heterogeneous mode unstable before the homogeneous mode. Its effects appear away from that branch, where weak coupling between equitable blocks preserves a combinatorial family of stable equilibria.

Several ingredients of our analysis are known separately. Equitable partitions have been used to reduce consensus and synchronization models \citep{schaub2016,siddique2018}. Network threshold models have been used to study cascades \citep{watts2002,wiedermann2020}, and weakly coupled cusp systems can support multiple attractors \citep{izhikevich1998}. Our contribution is to connect these ideas within Kaye's threshold model. We give a contraction condition valid for every row-stochastic network, show that the homogeneous mode loses stability before any heterogeneous mode, prove the existence of $r^m$ stable full-network equilibria under weak interblock coupling, and derive the location and block-activity contrast of the heterogeneous cusps.

Exact equitability is a mathematical assumption that makes the quotient reduction exact; we do not assume that empirically identified groups satisfy it. Network symmetries can produce equitable partitions in empirical networks \citep{macarthur2008}, but useful coarse partitions are often only approximately equitable \citep{kate2009,squillace2026}. Nathe et al. identify nearly equitable ``dynamical communities'' in weighted technological, biological, and social networks and show that they predict synchronized behaviour more effectively than conventional community-detection methods \citep{nathe2022}. Proposition \ref{prop:approx-equitable} gives a more limited result: a hyperbolic equilibrium associated with an exact equitable partition persists under sufficiently small perturbations of the interaction matrix, although node activities need no longer be equal within each block and the exact quotient reduction is lost.

The cusp asymptotics quantify how much interblock coupling an unequal-activity state can withstand near the onset of scalar bistability. If $\delta=\beta-\beta_c$, the coupling at the heterogeneous cusp is of order $\delta$, whereas the difference between the block activities is of order $\delta^{1/2}$. Thus the critical coupling approaches zero faster than the activity difference: close to $\beta_c$, the two blocks can still have visibly different activity levels even though a small increase in interblock coupling makes that state disappear.

The assumptions limit the scope of the results. Row-stochasticity fixes the total incoming feedback weight at one for every node, allowing the distribution of that weight among neighbours to vary without changing its total magnitude. If the row sums differ, the homogeneous set need not remain invariant and topology-dependent instabilities may occur. Equitable partitions give exact quotient systems. Proposition \ref{prop:approx-equitable} shows that individual hyperbolic equilibria persist under small violations of equitability, but it does not give an approximately invariant quotient or uniform error bounds along a complete bifurcation diagram. Finally, \eqref{eq:network-model} is a deterministic response model and is not derived from a stochastic network of binary threshold agents. Such a derivation would require a separate limit theorem.

Two questions remain. First, can uniform error bounds be obtained for approximately equitable networks away from bifurcations? Second, how do asymmetric quotient matrices change the pair of heterogeneous cusps found for symmetric two-block coupling? Both questions preserve the distinction established here: the threshold mechanism produces homogeneous hysteresis, while the network structure supports additional heterogeneous stable states.

\section{Conclusion}

We studied how Kaye's threshold mechanism changes when feedback is transmitted through a row-stochastic network. If $\beta\sup_z f(z)<1$, the network has a unique globally attracting equilibrium. For every feedback strength, the homogeneous dynamics reproduce Kaye's scalar equation, so the network changes neither its folds nor its cusp, and no heterogeneous mode becomes unstable first. Equitable partitions add a different effect. When the scalar equation has several stable equilibria, independent assignments of those equilibria to the blocks persist under sufficiently weak interblock coupling and are stable in the full network. For two symmetrically coupled blocks, branches with unequal activities terminate at heterogeneous cusp bifurcations. Near the onset of scalar bistability, the coupling at these cusps, the difference between the block activities, and the shift in the external stimulus obey the scaling laws derived in Theorem \ref{thm:heterocusp}. These results distinguish the homogeneous hysteresis inherited from Kaye's model from the additional multistability supported by the equitable block structure.

\section*{Competing interests}

The author declares no competing interests.

\appendix

\section{Explicit fold and cusp equations}
\label{app:fold-cusp}

For completeness, write
\[
 A=\frac{f(a+s)+f(a-s)}2,
 \qquad
 C=\frac{f(a+s)-f(a-s)}2,
\]
\[
 P=\frac{f'(a+s)+f'(a-s)}2,
 \qquad
 Q=\frac{f'(a+s)-f'(a-s)}2.
\]
Removing the nonzero denominator from $\det DP_\beta$ gives the fold equation
\begin{equation}
 R(a,s;\beta)
 =\beta s C^2-(1-\beta A)(D-sA)=0.
 \label{eq:fold-numerator}
\end{equation}
Its derivatives are
\begin{align}
 R_a&=2\beta sCQ+\beta P(D-sA)-(1-\beta A)(C-sP),\\
 R_s&=\beta C^2+2\beta sCP+\beta Q(D-sA)+(1-\beta A)sQ.
\end{align}
Since a kernel direction is proportional to $(D-sA,sC)$, the cusp equation used in the computations is
\begin{equation}
 (D-sA)R_a+sCR_s=0.
 \label{eq:cusp-numerator}
\end{equation}
Equations \eqref{eq:fold-numerator} and \eqref{eq:cusp-numerator} were solved for $(a,s)$ at fixed $\beta$, after which \eqref{eq:parameter-map} was used to recover $(u,\eps)$.

\section{Independent Lyapunov--Schmidt check}
\label{app:LS}

Let $\mathcal G(x,u,\eps)=0$ denote the two-block equilibrium equation, and let $A=D_x\mathcal G$ at the candidate cusp. Choose a right null vector $v$ with $\|v\|_2=1$ and $v_1>0$, and scale a left null vector $w$ so that $w^Tv=1$. With $\mathcal B=D_x^2\mathcal G$ and $\mathcal C=D_x^3\mathcal G$, the cusp conditions and cubic coefficient are
\[
 \det A=0,
 \qquad
 w^T\mathcal B[v,v]=0,
\]
\begin{equation}
 c=w^T\bigl(\mathcal C[v,v,v]+3\mathcal B[v,h_{20}]\bigr),
 \qquad
 Ah_{20}=-\mathcal B[v,v],
 \quad w^Th_{20}=0.
 \label{eq:cubic-coefficient}
\end{equation}
For a parameter $p\in\{u,\eps\}$, let
\[
 Ah_p=-(I-vw^T)\mathcal G_p,
 \qquad w^Th_p=0,
\]
and define
\[
 g_p=w^T\mathcal G_p,
 \qquad
 g_{sp}=w^T\bigl(D_x\mathcal G_p\,v+\mathcal B[v,h_p]\bigr).
\]
The unfolding is transversal when
\begin{equation}
 \det
 \begin{pmatrix}
 g_u&g_\eps\\
 g_{su}&g_{s\eps}
 \end{pmatrix}
 \ne0.
 \label{eq:unfolding}
\end{equation}
For the benchmark cusp, $c=32.762588$ and the determinant in \eqref{eq:unfolding} is $2.355585$.

\bibliographystyle{unsrtnat}
\bibliography{references}

\end{document}